\documentclass[11pt]{article}
\usepackage[margin=1in, headheight=22.54448pt]{geometry}
\usepackage[bottom, flushmargin]{footmisc}
\RequirePackage{fix-cm}
\usepackage{fancyhdr}
\usepackage{multirow, multicol}
\usepackage{booktabs}
\usepackage{dsfont,amssymb,amsmath,amsthm,centernot,graphicx,accents}
\usepackage{enumitem}
\usepackage{natbib}
\usepackage[colorlinks = true, citecolor = blue]{hyperref}
\usepackage{cleveref}
\usepackage{setspace}
\usepackage{color} 
\usepackage{xcolor} 
\usepackage[edges]{forest}
\usepackage[outline]{contour}
\usepackage{booktabs}

\hypersetup{
	colorlinks,
	linkcolor= black, 
	urlcolor=black,  
	citecolor= black 
}

\usepackage{tikz}
\usepackage[normalem]{ulem}
\usepackage{pgfplots}
\usepackage{pgfplotstable}
\usepackage{setspace}

\usepackage{multirow}
\usepackage{multicol}
\usepackage{array}

\renewcommand{\cite}[1]{\citeyear{#1}}
\theoremstyle{plain}

\newtheorem{theorem}{Theorem}
\newtheorem*{theorem*}{Theorem}
\Crefname{theorem}{Theorem}{Theorems}
\crefname{theorem}{theorem}{theorems}

\newtheorem*{proposition*}{Proposition}
\Crefname{proposition}{Proposition}{Propositions}
\crefname{proposition}{proposition}{propositions}

\newtheorem{corollary}{Corollary}
\newtheorem*{corollary*}{Corollary}
\Crefname{corollary}{Corollary}{Corollaries}
\crefname{corollary}{corollary}{corollaries}

\newtheorem{lemma}{Lemma}
\newtheorem*{lemma*}{Lemma}
\Crefname{lemma}{Lemma}{Lemmas}
\crefname{lemma}{lemma}{lemmas}

\newtheorem*{conjecture*}{Conjecture}
\Crefname{conjecture}{Conjecture}{Conjectures}
\crefname{conjecture}{conjecture}{conjectures}

\theoremstyle{definition}

\newtheorem*{definition*}{Definition}
\Crefname{definition}{Definition}{Definitions}
\crefname{definition}{definition}{definitions}

\newtheorem{assumption}{Assumption}
\newtheorem*{assumption*}{Assumption}
\Crefname{assumption}{Assumption}{Assumptions}
\crefname{assumption}{assumption}{assumptions}

\newcommand{\indep}{\perp \!\!\! \perp}

\long\def\symbolfootnote[#1]#2{\begingroup%
\def\thefootnote{\fnsymbol{footnote}}\footnote[#1]{#2}\endgroup}

\usepackage{footnote}
\makesavenoteenv{tabular}
\usepackage{fancyhdr}
\newcommand{\documenttitle}{Thesis}

\renewcommand{\max}{\operatornamewithlimits{max}}

\pgfplotsset{compat=1.15}

\title{Design-Based Multi-Way Clustering}
\author{Luther Yap}
\date{\today}

\begin{document}
\pagenumbering{arabic}
\maketitle

\begin{abstract}
This paper extends the design-based framework to settings with multi-way cluster dependence, and shows how multi-way clustering can be justified when clustered assignment and clustered sampling occurs on different dimensions, or when either sampling or assignment is multi-way clustered. Unlike one-way clustering, the plug-in variance estimator in multi-way clustering is no longer conservative, so valid inference either requires an assumption on the correlation of treatment effects or a more conservative variance estimator. Simulations suggest that the plug-in variance estimator is usually robust, and the conservative variance estimator is often too conservative.  
\end{abstract}


\section{Introduction}

There have been thousands of papers that adjust standard errors for clustering in linear regression. To address the issue of when such clustering is appropriate, \citet{abadie2023should} extended the design-based framework of \citet{abadie2020sampling} to show how clustering on a particular dimension, such as state, is justified when there is clustered sampling or assignment on that dimension. However, their framework and theory apply only to one-way clustering, which leaves an open question on how papers that clustered on multiple dimensions, such those applying the plug-in variance of \citet{cameron2011robust}, henceforth CGM, can be justified. 
There are situations where sampling and assignment occur on different clustering dimensions (e.g., \citet{gruber1995health}). 
There are also sampling (e.g., \citet{hersch1998compensating}) or assignment (e.g., \citet{nunn2011slave}) mechanisms that are multi-way clustered. 
Then, there is an open question of whether the results of \citet{abadie2023should} generalize to multi-way clustering settings. This paper fills the gap. 

Until recently, accounting for the large sample behavior of design-based settings with multi-way clustering has been a difficult problem. 
Asymptotic theory for variables that have multi-dimensional dependence has thus far relied on separate exchangeability (e.g., \citet{davezies2018asymptotic}). 
Separate exchangeability implies that the marginal distributions of clusters are exchangeable. \citep{mackinnon2021wild}
However, by construction, separate exchangeability is violated in a design-based framework because the residual $u_i = W_i u_i(1) + (1-W_i) u_i(0)$ depends on nonstochastic potential residual $u_i(w)$ and treatment $W_i$. Hence, even if the treatments $W_i$'s were identically distributed over clusters, the marginal distribution of $u_i$ cannot be the same as the $W_i$'s are weighted differently. Hence, limit theory that accommodates heterogeneity of clusters and observations is required.
By building on the central limit theorem in \citet{yap2023general}, this paper obtains results on large-sample behavior of standard estimators in this environment. 

This paper shows that while most results in the one-way design-based setup generalize to multi-way clustering, there are some nuances in multi-way clustering. 
As we would expect, when sampling and assignment occur on different clustering dimensions, it is necessary to cluster on both dimensions to obtain valid inference. 
When sampling is two-way clustered, it is necessary to cluster on both dimensions for valid inference. 
When assignment is one-way clustered, \citet{abadie2023should} showed that the standard \citet{liang1986longitudinal} plug-in variance estimator is conservative, even though it is still necessary to cluster in some way for valid inference.
Similarly, in several data-generating processes with multi-way clustered assignment, the CGM plug-in estimator is conservative. 
However, unlike one-way clustering, the CGM  estimator is no longer always conservative. 
In fact, it is possible to construct a data-generating process that makes the CGM variance estimator anticonservative when there is multiway sampling. 

In response to the anticonservativeness of the CGM estimator in design-based settings, there are two approaches that empirical researchers may take. The first approach is to make an assumption on how the individual treatment effects are correlated within the same cluster: CGM is conservative when the correlation is positive, which is reasonable in most applications. The second approach is to remain agnostic and to use CGM2, a more conservative version of the CGM variance estimator proposed by \citet{davezies2018asymptotic}. Since the simulations show that CGM2 is often unnecessarily conservative, making an assumption on the correlation is usually the more reasonable approach.

\section{Setting}
We have a sequence of populations indexed by $k$, each with $n_k$ units. Each unit is indexed by $i = 1, \cdots , n_k$, partitioned into clusters on two dimensions $G,H$, indexed by $g$ and $h$. Let $m=(g,h)$ denote the intersection of two cluster indices that are nonempty. 
Let $g_{ki}, h_{ki}$ denote the cluster that unit $i$ belongs to on the respective dimensions, and $m_{ki}$ its intersection of clusters. For treatment variable $W \in \{ 0,1 \}$, we have the potential outcome is denoted $y_{ki} (w)$ that is nonstochastic.

We are interested in the population average treatment effect (ATE):
\begin{align*}
	\tau_k &= \frac{1}{n_k} \sum_{i=1}^{n_k} (y_{ki} (1) - y_{ki} (0)) 
\end{align*}

With $\alpha_k := (1/n_k)\sum_{i=1}^{n_k} y_{ki}(0)$ and residuals $U_{ki} := Y_{ki} - \alpha_k - \tau_k W_{ki}$, the potential residuals are denoted:
\begin{align*}
    u_{ki}(1) &= y_{ki}(1) - (\alpha_k + \tau_k) \\
    u_{ki}(0) &= y_{ki}(0) - \alpha_k
\end{align*}

Let $R_{ki}$ denote the sampling indicator, so observation $i$ in population $k$ is observed when $R_{ki}=1$. Similarly, $W_{ki}$ denotes the treatment indicator. Hence, for $R_{ki}=1$, we observe the triple $\{ Y_{ki}, W_{ki}, m_{ki} \}$, with $Y_{ki}= W_{ki} y_{ki}(1) + (1-W_{ki})y_{ki}(0)$. The high-level assumption on assignment and sampling is that $(R_{ki},W_{ki}) \indep (R_{kj},W_{kj})$ whenever they do not share any cluster, and that sampling and assignments are independent, as stated in Assumption \ref{asmp:indep}.

\begin{assumption} \label{asmp:indep}
    $(R_{ki},W_{ki}) \indep (R_{kj},W_{kj})$ if $g_{ki} \ne g_{kj}$ and $h_{ki} \ne h_{kj}$. Further, $\{ R_{ki} \}_{i=1}^{n_k} \indep \{ W_{ki} \}_{i=1}^{n_k}$. For all $i$, there exists some $b_{k1}$ and $b_{k0}$ such that $E\left[R_{ki}W_{ki}\right] = b _{k1}$ and $E\left[R_{ki}\left(1-W_{ki}\right)\right] =b_{k0}$.
\end{assumption}

$b_{k1}$ is the probability that an individual is observed and treated; $b_{k0}$ is the expected probability that an individual is observed and untreated. 

The above assumption can be generated from several design-based settings. One possibility is to adapt the data-generating process from \citet{abadie2023should} for one-way clustering, just that the sampling and assignment processes are allowed to occur on different clustering dimensions. This environment strictly generalizes their setting: their setting is a special case where both sampling and assignment occur on the same dimension, and everyone is in their own cluster on the other dimension. Another possibility is to have multi-way clustering occur on either the assignment or sampling dimension, and independence over units on the other dimension. I detail both possibilities in the next two subsections, and provide empirical examples.

\subsection{Sampling and Assignment on Different Dimensions}
Without loss of generality, suppose there is clustered sampling on the $G$ dimension. Here, every $G$ cluster is independently sampled with probability $q_k$. If a cluster is sampled, then units within the cluster are sampled independently with probability $p_k$. Clustered assignment can occur on the $H$ dimension. Every cluster on the $H$ dimension has an assignment probability $B_{kh}$, drawn from a distribution with mean $\mu_k$ and $\sigma_k^2$. Then, for a given cluster $h$, units within that cluster are independently assigned treatment with probability $B_{kh}$. 

An empirical example is \citet{gruber1995health}, who study the effect of health insurance on retirement. We may want to cluster by household and state-year cell. They use the CPS data which has a sample that is negligibly small relative to the superpopulation of households, and each household reveals data on employment in previous years (where they could have moved across different states). It then seems reasonable to believe that sampling occurs on the household dimension. Assignment occurs on the state dimension because insurance is affected by state policy. Since the same incumbent party has differing policies over states, and the same state has differing policies when the party differs over time, assignment to health insurance varies by state clusters. 

\subsection{Multiway Clustering on Assignment}
This subsection explains a possible mechanism for multiway clustering on the assignment dimension. Since the setup for multiway sampling is analogous, its exposition is omitted for brevity. 

Data is generated by independently drawing $A_{kg}\in[0,1]$, $B_{kh}\in[0,1]$ and $e_{ki}\sim U[0,1]$, with $W_{ki}=1\left\{ e_{i}<A_{g(i)}B_{h(i)}\right\}$. The random variables $A_{kg}$ and $B_{kh}$ have means $\mu_{Ak}, \mu_{Bk}$ and variances $\sigma_{Ak}^2, \sigma_{Bk}^2$ respectively. This process nests several cases. If assignment is one-way clustered, then we can simply set $A_{kg}=1$. A special case is where assignment occurs at the intersection level, and we need both dimensions $G$ and $H$ to be assigned treatment for the unit to be treated. Then, $A_{kg}, B_{kh} \in \{0,1 \}$. This mechanism also nests the assignment mechanism where a unit is treated whenever either of its clustering dimensions is treated. To see this, observe that the unit is untreated only if both its clusters are untreated, so we can simply switch the labels of treatment and non-treatment. 

As an example of multi-way assignment, \citet{nunn2011slave} studied the effect of slave trade on trust today. They clustered by ethnic groups and district. The unit of observation is an individual. Both these dimensions affect assignment because people of similar ethnicities have similar history of slave trade. District is relevant since people in the same location are more likely hit by slave trade. 

An example of multi-way sampling is from \citet{hersch1998compensating}, who was interested in the effect of injury risk on wages. She clustered by industry and occupation. When sampling from a large number of industries and occupations, we would observe an individual only when both his industry and his occupation were sampled.  

\section{Least Squares Estimator and Variance}

Let $N_k := \sum_{i=1}^{n_k} R_{ki}$ denote the number of sampled units. Further, $N_{k1} := \sum_{i=1}^{n_k} R_{ki} W_{ki}$ and $N_{k0} := \sum_{i=1}^{n_k} R_{ki} (1-W_{ki})$. The least squares estimator is:
\begin{align*}
    \hat{\tau}_k &= \frac{1}{N_{k1}} \sum_{i=1}^{n_k} R_{ki} W_{ki} Y_{ki} - \frac{1}{N_{k0}} \sum_{i=1}^{n_k} R_{ki} (1-W_{ki}) Y_{ki}
\end{align*}

The rest of this section first shows how the estimator is asymptotically normal, and derives an expression for its asymptotic variance. Then, it presents the asymptotic limits of various commonly-used variance estimators.

\subsection{Large Sample Properties of OLS Estimator} \label{sec:ols_estimator}
To state the main result, I first define a few terms. Let $\mathcal{N}_i$ denote the neighborhood of $i$, which is the set of observations that are plausibly correlated with $i$. In the multi-way clustering context, $\mathcal{N}_i := \{ j: g_{ki}=g_{kj} \} \cup \{ j: h_{ki}=h_{kj} \}$. For $C \in \{G ,H \}$, let $\mathcal{N}^C_c$ denote the set of observations in cluster $c$, and $N^C_c := |\mathcal{N}^C_c|$ denote the number of observations in cluster $c$ on the $C$ dimension. Further, define $\xi_{ki}$ as the demeaned residual for individual $i$ that features in the variance of $\hat{\tau}_k$:
\begin{align*}
\xi_{ki} &:= \frac{1}{b_{k1}}\left(R_{ki}W_{ki}-b_{k1}\right)u_{ki}(1)-\frac{1}{b_{k0}}\left(R_{ki}\left(1-W_{ki}\right)-b_{k0}\right)u_{ki}(0)
\end{align*}

\begin{assumption} \label{asmp:regularity}
    Let $Q_{ki}$ denote $R_{ki} W_{ki}$, $R_{ki} (1-W_{ki})$ or $\xi_{ki}$. For some positive $K_0<\infty$, $C\in \{ G,H \}$ and $\lambda_k = Var(\sum_{i=1}^{n_k} Q_{ki})$, the following hold:
    \begin{enumerate}
        \item $|y_{ki}(w)| \leq K_0$.
        \item $\frac{1}{\lambda_k} \max_c (N^C_c)^2 \rightarrow 0$.
        \item $\frac{1}{\lambda_k} \sum_c (N^C_c)^2 \leq K_0$.
    \end{enumerate}
\end{assumption}

These regularity conditions are required in the multi-way clustering context so that the central limit theorem (CLT) from \citet{yap2023general} can be applied. These regularity conditions require bounded potential outcomes, the contribution of the largest cluster to the total variance be small, and a summability condition. 

\begin{theorem} \label{thm:tau_distr}
    Under \Cref{asmp:indep} and \Cref{asmp:regularity},
    \begin{equation}
    \frac{\sqrt{N_k} (\hat{\tau}_k - \tau_k)}{\sqrt{v_k}} \xrightarrow{d} N(0,1)
    \end{equation}   
    where 
\begin{align*}
    v_k &:= \frac{N_k}{n_k^2} \sum_{i=1}^{n_k} \sum_{j \in \mathcal{N}_i} E[\xi_{ki} \xi_{kj}]  
\end{align*}
and
\begin{align*}
    E\left[\xi_{ki}\xi_{kj}\right] & =\left(\frac{1}{b_{k1}^{2}}E\left[R_{ki}R_{kj}\right]E\left[W_{ki}W_{kj}\right]-1\right)u_{ki}(1)u_{kj}(1)  \\
    &\qquad +\left(\frac{1}{b_{k0}^{2}}E\left[R_{ki}R_{kj}\right]E\left[\left(1-W_{ki}\right)\left(1-W_{kj}\right)\right]-1\right)u_{ki}(0)u_{kj}(0)\\
    & \qquad-\left(\frac{1}{b_{k1}b_{k0}}E\left[R_{ki}R_{kj}\right]E\left[W_{ki}\left(1-W_{kj}\right)\right]-1\right)u_{ki}(1)u_{kj}(0)  \\
    &\qquad-\left(\frac{1}{b_{k1}b_{k0}}E\left[R_{ki}R_{kj}\right]E\left[\left(1-W_{ki}\right)W_{kj}\right]-1\right)u_{ki}(0)u_{kj}(1)
\end{align*}  
\end{theorem}

In \Cref{thm:tau_distr}, the variance expression is a function of expectations of $R$ and $W$ cross products across different units. These cross-products have different expressions, depending on the environment that generated the data, as stated in Lemmas \ref{lem:diffRW} and \ref{lem:mwclusW} below.

\begin{lemma} \label{lem:diffRW}
    Suppose \Cref{asmp:indep} holds, there is clustered sampling on dimension $G$, and clustered assignment on dimension $H$. Then, $g_{ki} \ne g_{kj}$ implies $E\left[R_{ki}R_{kj}\right] = p_{k}^{2}q_{k}^{2}$ and $h_{ki} \ne h_{kj}$ implies $E\left[W_{ki}W_{kj}\right]=\mu^{2}$. If $g_{ki}=g_{kj}$, $E\left[R_{ki}R_{kj}\right]=q_{k}p_{k}^{2}$. If $h_{ki}=h_{kj}$,  $E\left[W_{ki}W_{kj}\right]= \sigma^{2}+\mu^{2}$, $E\left[W_{ki}\left(1-W_{kj}\right)\right] = \mu\left(1-\mu\right)-\sigma^{2}$, and $E\left[\left(1-W_{ki}\right)\left(1-W_{kj}\right)\right] =\left(1-\mu\right)^{2}+\sigma^{2}$.
\end{lemma}

\begin{lemma}\label{lem:mwclusW}
Suppose \Cref{asmp:indep} holds and there is multiway clustering in assignment. Then, if $m_{ki}=m_{kj}$, then $E\left[W_{ki}W_{kj}\right]=\left(\mu_{A}^{2}+\sigma_{A}^{2}\right)\left(\mu_{B}^{2}+\sigma_{B}^{2}\right)$. If $g_{ki}=g_{kj},h_{ki} \ne h_{kj}$, then $E\left[W_{ki}W_{kj}\right]=\left(\mu_{A}^{2}+\sigma_{A}^{2}\right)\mu_{B}^{2}$. If $g_{ki}\ne g_{kj}$ and $h_{ki} \ne h_{kj}$, then $E\left[W_{ki}W_{kj}\right]=\mu_{A}^{2}\mu_{B}^{2}$. 
\end{lemma}

Multi-way sampling can be treats $R$ analogously to $W$ in Lemma \ref{lem:mwclusW}.

\subsection{Variance Estimators} \label{sec:variance_estimators}
The true variance as stated above is a function of potential residuals, so a direct plug-in is infeasible. Variance estimators hence use a feasible analog. Define:
\begin{align*}
    \eta_{ki} & :=R_{ki}\left(\frac{W_{ki}}{b_{k1}}-\frac{1-W_{ki}}{b_{k0}}\right)U_{ki} \\
    \hat{\eta}_{ki} & := R_{ki}\left(\frac{W_{ki}}{\hat{b}_{k1}}-\frac{1-W_{ki}}{\hat{b}_{k0}}\right)\hat{U}_{ki} \\
    \hat{U}_{ki} &:= Y_{ki} - \hat{\alpha}_k -\hat{\tau}_k W_{ki}
\end{align*}

Observe that $\xi_{ki} = \eta_{ki} - E[\eta_{ki}]$, and $E[\eta_{ki}] = u_{ki}(1)-u_{ki}(0)$. The common variance estimators include the heteroskedasticity-robust variance estimator attributed to \citet{eicker1967limit}, \citet{huber1967under} and \citet{white1980heteroskedasticity} (EHW), and the one-way cluster robust variance estimator attributed to \citet{liang1986longitudinal} (LZ). The CGM estimator \citep{cameron2011robust} is the sum of LZ estimators on the two different dimensions minus the LZ estimator of the intersection to avoid double-counting terms. The CGM2 estimator \citep{davezies2018asymptotic} uses the sum without subtracting the intersection so terms are double-counted. To be precise,
\begin{align*}
    \hat{V}_{EHW} &:= \frac{1}{N_k} \sum_{i=1}^{n_k} \hat{\eta}_{ki}^2 \\
    \hat{V}_{LZC} &:= \frac{1}{N_k} \sum_{i=1}^{n_k} \sum_{j \in \mathcal{N}^C_{c_{ki}}} \hat{\eta}_{ki} \hat{\eta}_{kj} \\
    \hat{V}_{CGM} &:= \frac{1}{N_k} \sum_{i=1}^{n_k} \sum_{j \in \mathcal{N}_i} \hat{\eta}_{ki} \hat{\eta}_{kj} = \hat{V}_{LZG} + \hat{V}_{LZH} - \hat{V}_{LZM} \\
    \hat{V}_{CGM2} &:= \frac{1}{N_k} \sum_{i=1}^{n_k} \sum_{j \in \mathcal{N}_i} \hat{\eta}_{ki} \hat{\eta}_{kj} + \frac{1}{N_k} \sum_{i=1}^{n_k} \sum_{j \in \mathcal{N}^M_{m_{ki}}} \hat{\eta}_{ki} \hat{\eta}_{kj} = \hat{V}_{LZG} + \hat{V}_{LZH}
\end{align*}
where, $\hat{V}_{LZG}$ and $\hat{V}_{LZH}$ denote the LZ estimator on the G and H dimensions respectively, while $\hat{V}_{LZM}$ is the LZ estimator treating each intersection $m$ as a single cluster. 

\begin{assumption} \label{asmp:eta}
    For $C, C^{\prime} \in \{ G,H \}$, let $\lambda^C_k := \sum_{i=1}^{n_k} \sum_{j \in \mathcal{N}^{C}_{c_{ki}}} E[\eta_{ki} \eta_{kj}]$. Then, $\frac{1}{\lambda^C_k} \max_c (N^{C^\prime}_c)^2 \rightarrow 0$ and $\frac{1}{\lambda^C_k} \sum_c (N^{C^\prime}_c)^2 \leq K_0$.
\end{assumption}

The conditions in Assumption \ref{asmp:eta} are required so that the asymptotic error incurred by using $\hat{V}$ relative to the true $V$ converges to zero. Since the strategy for showing such convergence is similar to \citet{yap2023general}, an analogous summability condition and a condition on the largest cluster having a negligible contribution to the variance are required.

\begin{theorem} \label{thm:var_converge}
    Under \Cref{asmp:indep}, \Cref{asmp:regularity} and \Cref{asmp:eta}, for $LZC \in \{ LZG,LZH \}$,
    \begin{align*}
        \frac{\hat{V}_{EHW}}{V_k^{EHW}} \xrightarrow{p} 1 \qquad \frac{\hat{V}_{LZC}}{V_k^{LZC}} \xrightarrow{p} 1 \qquad \frac{\hat{V}_{CGM}}{V_k^{CGM}} \xrightarrow{p} 1 \qquad \frac{\hat{V}_{CGM2}}{V_k^{CGM2}} \xrightarrow{p} 1
    \end{align*}
    where
    \begin{align*}
        V_k^{EHW} &:= \frac{1}{N_k} \sum_{i=1}^{n_k} E[\eta_{ki}^2] \\
        V_k^{LZC} &:= \frac{1}{N_k} \sum_{i=1}^{n_k} \sum_{j \in \mathcal{N}^C_{c_{ki}}} E[\eta_{ki} \eta_{kj}] \\
        V_k^{CGM} &:= \frac{1}{N_k} \sum_{i=1}^{n_k} \sum_{j \in \mathcal{N}_i} E[\eta_{ki} \eta_{kj}] \\
        V_k^{CGM2} &:= \frac{1}{N_k} \sum_{i=1}^{n_k} \sum_{j \in \mathcal{N}_i} E[\eta_{ki} \eta_{kj}] + \frac{1}{N_k} \sum_{i=1}^{n_k} \sum_{j \in \mathcal{N}^M_{m_{ki}}} E[\eta_{ki} \eta_{kj}] \\
        E[\eta_{ki} \eta_{kj}] &= E[\xi_{ki} \xi_{kj}] + E[\eta_{ki}] E[\eta_{kj}]
    \end{align*}
\end{theorem}

The result allows us to observe differences in variance estimators in Corollary \ref{cor:var_diff}.

\begin{corollary} \label{cor:var_diff}
The differences between the limit of the variance estimators and the true variance are as follows:
\begin{align*}
    V_{k}^{CGM}-v_{k} & =\frac{N_{k}}{n_{k}^{2}}\sum_{i}\sum_{j\in\mathcal{N}_{i}}E\left[\eta_{ki}\right]E\left[\eta_{kj}\right] \\
 V_{k}^{EHW}-v_{k} & =\frac{N_{k}}{n_{k}^{2}}\sum_{i}E\left[\eta_{ki}\right]^{2}-\frac{N_{k}}{n_{k}^{2}}\sum_{i}\sum_{j\in\mathcal{N}_{i}\backslash\left\{ i\right\} }E\left[\xi_{ki}\xi_{kj}\right] \\
 V_{k}^{LZG}-v_{k} & =-\frac{N_{k}}{n_{k}^{2}}\sum_{i}\sum_{j\in\mathcal{N}_{i}\backslash\mathcal{\mathcal{N}}_{g(i)}^{G}}E\left[\xi_{ki}\xi_{kj}\right]+\frac{N_{k}}{n_{k}^{2}}\sum_{i}\sum_{j\in\mathcal{\mathcal{N}}_{g(i)}^{G}}E\left[\eta_{ki}\right]E\left[\eta_{kj}\right] \\
 V_{k}^{CGM2}-v_{k} & =\frac{N_{k}}{n_{k}^{2}}\sum_{i}\sum_{j\in\mathcal{\mathcal{N}}_{m(i)}^{M}}E\left[\xi_{ki}\xi_{kj}\right]+\frac{N_{k}}{n_{k}^{2}}\sum_{i}\sum_{j\in\mathcal{\mathcal{N}}_{h(i)}^{H}}E\left[\eta_{ki}\right]E\left[\eta_{kj}\right]+\frac{N_{k}}{n_{k}^{2}}\sum_{i}\sum_{j\in\mathcal{\mathcal{N}}_{g(i)}^{G}}E\left[\eta_{ki}\right]E\left[\eta_{kj}\right]
\end{align*}
\end{corollary}

In general, we cannot guarantee that $V_k^{CGM}-v_k$ is positive. The simulation presents one example where this difference is negative. 
Nonetheless, observe that, since $u_{ki}(1) - u_{ki}(0) = E[\eta_{ki}]$, 
\begin{align*}
    u_{ki}(1) - u_{ki}(0) &= [y_{ki}(1) - (\alpha_k + \tau_k)] - [y_{ki}(0) - \alpha_k] \\
    &= y_{ki}(1) - y_{ki}(0) - \tau_k \\
    &=: \tau_{ki} - \tau_k
\end{align*}
Hence, $V_{k}^{CGM}-v_{k}$ will be positive whenever $\sum_{j \in \mathcal{N}_i} (\tau_{ki} - \tau_k) (\tau_{kj} - \tau_k) \geq 0$ i.e., when the sum of the correlation in the treatment effects for observations that are plausibly correlated is positive. This assumption is implied by the treatment effects of all units who are plausibly correlated having a positive correlation (i.e., $(\tau_{ki} - \tau_k) (\tau_{kj} - \tau_k) \geq 0$ for all $j \in \mathcal{N}_i$). However, the requirement for $V_k^{CGM} - v_k$ is weaker than this uniform assumption on the correlation, because we can have some negative correlations --- we merely require their overall sum to be weakly positive.
The assumption that $\sum_{j \in \mathcal{N}_i} (\tau_{ki} - \tau_k) (\tau_{kj} - \tau_k) \geq 0$ is reasonable in many empirical situations, as we would usually expect treatment effects of neighbors to be positively rather than negatively correlated. However, this assumption cannot be written in terms of $R_{ki}$ and $W_{ki}$ that our design-based setting has restrictions on. Hence, nothing in the design-based setting necessarily implies that CGM is conservative.

The asymptotic underestimation of the EHW estimator comes from failing to account for correlations in $\xi$, but it may be compensated by the nonzero mean of $\eta$. 
A similar interpretation is obtained in that underestimation comes from the failure to account for correlations in H that are not part of the g cluster. This underestimation may be offset by $\sum_{i}\sum_{j\in\mathcal{\mathcal{N}}_{g(i)}^{G}}E\left[\eta_{ki}\right]E\left[\eta_{kj}\right] \geq 0$, which is weakly positive because the object can be written as a sum of outer products. 

Finally, it can be shown that CGM2 is conservative. 
\begin{corollary} \label{cor:cgm2pos}
$V_{k}^{CGM2}-v_{k} \geq 0$.
\end{corollary}

The CGM2 difference above is guaranteed to be positive, because it can be written as an additive function of one-way clustered objects. Its conservativeness arises from double-counting the intersection of clusters, and the nonzero mean of $\eta$. 

Consequently, in contrast to the plug-in (LZ) estimator in one-way clustering being conservative, the plug-in (CGM) estimator in two-way clustering is no longer conservative. This leaves researchers with two options: they can either assume that $\sum_{j \in \mathcal{N}_i} (\tau_{ki} - \tau_k) (\tau_{kj} - \tau_k) \geq 0$, or they can use CGM2.

\section{Simulations}
This section describes and presents results for simulations to illustrate how the theoretical results perform in practice. The simulations use the following procedure:
\begin{enumerate}
    \item Construct a population of size $n_k$. Each unit in the population is a tuple of $(g(i),h(i),Y_{ki}(1),Y_{ki}(0))$.
    \item For every simulation $s$:
    \begin{enumerate}
        \item Draw random variables $(R_{ki},W_{ki})$ for every unit and hence construct a sample of observations with $R_{ki}=1$. For every unit in the sample, $(g(i),h(i),W_{ki},Y_{ki})$ is observed.
        \item Calculate $\hat{\tau}$ from the regression.
        \item Calculate the variance estimates for various procedures and construct a confidence interval to determine if the CI covers $\tau$. 
    \end{enumerate}
    \item Aggregate the coverage over the simulations.
\end{enumerate}

Then, every design modifies some part of the procedure. The following parts can be modified:
\begin{enumerate}
    \item How clusters $(g,h)$ are assigned in the population. 
    \begin{enumerate}
        \item Balanced: Use $n_k=1e6$, with $1000$ $G$ clusters and $1000$ $H$ clusters, so there is one observation for every intersection.
        \item Staircase: Let $k_{1}$ denote an odd number, and $(g,h)$ describe a cluster intersection that is in cluster $g$ on the $G$ dimension and in cluster $h$ on the $H$ dimension. In the population, individuals can only belong to cluster intersections of the form $(k_{1},k_{1})$, $(k_{1},k\pm1)$ and $(k_{1}\pm1,k_{1})$. The mass of the population in $(k_{1},k_{1})$ is four times larger than cluster intersections of the other four forms, which are of equal mass. 
    \end{enumerate}
    \item How potential outcomes and hence $\tau_{ki}$ are constructed in the population. Use $u_i \sim N(0,0.1)$ so that $Y_{ki} = \tau_{ki}W_{ki} + u_i$.
    \begin{enumerate}
        \item Gvar: $\tau_{ki}=\tau_{g(i)} + \tau_{h(i)}$, where $\tau_g=\pm2$ with equal probability and $\tau_h=\pm1/2$ with equal probability. 
        \item Hvar: $\tau_{ki}=\tau_{g(i)} + \tau_{h(i)}$, where $\tau_h=\pm2$ with equal probability and $\tau_g=\pm1/2$ with equal probability.
        \item same: $\tau_{ki}=\tau_{g(i)} + \tau_{h(i)}$, where $\tau_g=\pm1$ with equal probability and $\tau_h=\pm1$ with equal probability.
        \item constant: $\tau_{ki}=1$. 
        \item oddeven: $\tau_{ki}=1$ if $g(i)$ and $h(i)$ are both odd, and $-1$ otherwise. 
    \end{enumerate}
    \item How $R_{ki}$'s are generated for every sample.
    \begin{enumerate}
        \item $p_k$ individual sampling probability
        \item $q_k$ cluster sampling probability for $G$
        \item Multiway Sampling: The intersection $(g,h)$ is sampled when $A^{sam}_{kg}=1$ and $B^{sam}_{kh}=1$, each independently drawn from $Be(0.25)$. Then, $p_k=0.25$.
    \end{enumerate}
    \item How $W_{ki}$'s are generated for every sample. 
    \begin{enumerate}
        \item Hway: Assignment probability $B_{kh}$ for every cluster is drawn from $U[0,1]$. $W_{ki} = 1$ with probability $B_{kh(i)}$.
        \item none: $W_{ki}=1$ with probability 1/2 so $W_{ki}$ is iid.
        \item AND: $A_{kg} \in\{0,1\}$ and $B_{kh}\in\{0,1\}$ are drawn independently, and each takes the value 1 with probability $1/\sqrt{2}$. Then, $W_{ki}=A_{kg(i)}B_{kh(i)}$.
    \end{enumerate}
\end{enumerate}

The staircase design for constructing the clusters is artificially designed so that $V_{CGM,k}-v_k\leq0$. Details are in Appendix \ref{sec:cgm_counter}. I use $nsim=5000$ simulation draws and report eight informative designs. 

\begin{table}
    \centering
    \footnotesize
    \caption{Comparison of Standard Methods for Different Designs} \label{tab:sim}
    
\begin{tabular}{lrrrrrrrrrr}
\toprule
&EHWCov & LZGCov & LZHCov & CGMCov & CGM2Cov & EHWVar & LZGVar & LZHVar & CGMVar & CGM2Var\\
\midrule
D1&0.7736 & 0.9880 & 0.9884 & 0.9990 & 0.9994 & 4e-04 & 0.0018 & 0.0018 & 0.0032 & 0.0036\\
D2&0.7836 & 0.8518 & 0.9986 & 0.9994 & 0.9996 & 8e-04 & 0.0011 & 0.0064 & 0.0067 & 0.0076\\
D3&0.3258 & 0.8802 & 0.8790 & 0.9680 & 0.9706 & 4e-04 & 0.0054 & 0.0054 & 0.0103 & 0.0107\\
D4&0.2542 & 0.9200 & 0.9336 & 0.9874 & 0.9882 & 2e-04 & 0.0050 & 0.0054 & 0.0103 & 0.0104\\
\addlinespace
D5&0.9562 & 0.9502 & 0.9540 & 0.9476 & 0.9946 & 0e+00 & 0.0000 & 0.0000 & 0.0000 & 0.0000\\
D6&0.3000 & 0.9608 & 0.9898 & 0.9988 & 0.9990 & 1e-04 & 0.0025 & 0.0040 & 0.0065 & 0.0066\\
D7&0.9902 & 1.0000 & 0.9966 & 1.0000 & 1.0000 & 9e-04 & 0.0048 & 0.0012 & 0.0051 & 0.0060\\
D8&0.2434 & 0.9568 & 0.9556 & 0.9396 & 0.9950 & 3e-04 & 0.0143 & 0.0144 & 0.0125 & 0.0287\\
\bottomrule
\end{tabular}

    \begin{tabular}{lcccc}
\toprule
  & (g,h) & $\tau$ & R & W \\
\midrule
D1 & balanced & same & $q_k=1,p_k=1$ & AND \\
D2 & balanced & Hvar & $q_k=1,p_k=1$ & AND \\
D3 & balanced & same & multiway & none \\
D4 & balanced & Hvar & $q_k=0.05,p_k=1$ & Hway \\
D5 & balanced & constant & $q_k=1,p_k=1$ & AND \\
D6 & balanced & Hvar & $q_k=0.1,p_k=1$ & none \\
D7 & balanced & Gvar & $q_k=1,p_k=1$ & Hway \\
D8 & staircase & oddeven & multiway & none \\
\bottomrule
\end{tabular} \\
    \flushleft
    Notes: The first five columns denotes the mean coverage when using the various variance estimators across $nsim=5000$ simulations for a 95\% confidence interval. The latter five columns denote the mean variance. When $q_k=p_k=1$, to reduce the size of each sample, I first draw 1\% of the generated population and fix this draw as the new population. Then, every simulation ``draw" only reassigns the treatment status of every unit in this new population. 
\end{table}

Table \ref{tab:sim} reports the results for various designs (D). D1 shows that there is massive over-coverage in two-way assignment. However, by virtue of clustered assignment, if we used the EHW standard errors, we would not have correct coverage. D2 shows how $\tau$ matters. Even though there is multiway assignment, because there is more variation in $\tau$ on the $H$ dimension, it suffices to cluster on $H$. D3 shows that when there is multiway sampling, we need to cluster on both dimensions to have valid confidence intervals --- doing so on just one dimension will under-cover. D4 reflects some necessity for considering sampling and assignment. When sampling and assignment occur on different clustering dimensions, if we had just clustered on either one, we would not have 0.95 coverage, so we need to cluster on both.

D5 is an example when CGM yields correct coverage, namely, when $\tau$ is constant. D6 shows that when there is clustered sampling on $G$ but no clustered assignment, it suffices to cluster on $G$. If we had clustered on $H$ too, the variance would have been twice as large. D7 shows that there is no need to cluster on $G$ when there is clustered assignment only on $H$. The variance is four times as large when clustering on the unnecessary dimension. D8 provides the counterexample where CGM under-covers when there is clustered sampling. The details of this construction is given in Appendix \ref{sec:cgm_counter}.

Overall, the simulations illustrate the theoretical results that sampling and assignment matter. We can get exact coverage with multiway sampling or when treatment effects are constant. CGM can under-cover in multi-way sampling, while CGM2 is always conservative. In fact, the conservativeness of CGM2 can be quite substantial: CGM2 variance is five times what is required for valid inference in D7, for instance. While D8 shows how CGM may be anticonservative, it is usually possible to rule out such designs by institutional details or some economic model. 

Design-based inference fundamentally relies on the institutional knowledge of how sampling and assignment occur. This knowledge then informs what dimensions the researcher should cluster on, if at all. When moving from one-way clustering to multi-way clustering, this paper has shown how such demands on institutional knowledge still apply. To use the plug-in variance estimator in multi-way settings, however, researchers need to make an additional assumption on the correlation of treatment effects within clusters. Without such an assumption, valid inference may require a variance estimator that is unnecessarily conservative.



\clearpage

\appendix

\section{CGM Counterexample} \label{sec:cgm_counter}

This section theoretically constructs a counterexample such that the CGM estimator is anticonservative. Due to Corollary \ref{cor:var_diff}, $V_k^{CGM} - v_k = (N_k/ n_k^2) \sum_i \sum_{j \in \mathcal{N}_i} E[\eta_{ki}] E[\eta_{kj}]$, and we have shown that $E[\eta_{ki}] = \tau_{ki} -\tau_k$. Hence, for a counterexample where $V_k^{CGM} - v_k <0$, it suffices that $\sum_i \sum_{j \in \mathcal{N}_i} (\tau_{ki}- \tau_k) (\tau_{kj} - \tau_k) <0$.

Let $k_1$ denote an odd number, and $(g,h)$ describe a cluster intersection that is in cluster $g$ on the $G$ dimension and in cluster $h$ on the $H$ dimension. Suppose there are $M$ $G$ clusters and $M$ $H$ clusters, where $M$ is even. Let $M_0$ be a fixed number.
In the population, individuals can only belong to cluster intersections of the form $(k_1,k_1)$, $(k_1, k \pm 1)$ and $(k_1 \pm 1, k_1)$. This assumption implies that there cannot be individuals in cluster intersections where both cluster indexes are even, or if their difference is more than one. 
There are $4M_0$ observations in $(k_1,k_1)$ clusters and $M_0$ observations in other cluster intersections that are nonempty. Hence, the total number of observations is $n_k = 4 M_0M$. 

\begin{table}
    \centering
    \caption{Distribution of population with bounded cluster sizes} \label{tab:tau_distr}
    \begin{tabular}{c|ccc}
Type & Proportion & (g,h) & $\tau_{ki}-\tau_k$  \\
\midrule
1 & 1/2  &  $(k_1, k_1)$  &  1 \\
2 & 1/4 & $(k_1, k_1 \pm 1)$ & -1 \\
3 & 1/4 & $(k_1 \pm 1, k_1)$ & -1 \\
\hline
\end{tabular}
\end{table}

The $\tau_{ki}-\tau_k$ values for observations belonging to the various cluster intersections are given by Table \ref{tab:tau_distr}. Then, I show how this particular construction results in $\sum_i \sum_{j \in \mathcal{N}_i} (\tau_{ki}- \tau_k) (\tau_{kj} - \tau_k) <0$.

First, consider an observation $i$ in $(k_1,k_1)$ --- they account for half the observations. Then, $|\mathcal{N}_i|=8M_0$, where $4M_0$ of those observations are in $(k_1,k_1)$ and the remaining $4M_0$ are either in $(k_1\pm1,k_1)$ or $(k_1,k_1\pm1)$. Hence, $\sum_{j \in \mathcal{N}_i} (\tau_{ki}- \tau_k) (\tau_{kj} - \tau_k) = (1) (4-1-1-1-1) =0$. 

Next, consider an observation $i$ in $(k_1,k_1\pm1)$. The treatment of $(k_1 \pm 1,k_1)$ is identical. The units in either $(k_1,k_1 \pm 1)$ or $(k_1 \pm1,k_1)$ account for the other half of the units. Here, $|\mathcal{N}_i| = 7M_0$. For instance, for some $(k_1,k_1+1)$, $4M_0$ of those are in $(k_1,k_1)$ intersections, $M_0$ are in $(k_1,k_1+1)$, $M_0$ are in $(k_1,k_1-1)$ and $M_0$ are in $(k_1+2,k_1+1)$. Then, $\sum_{j \in \mathcal{N}_i} (\tau_{ki}- \tau_k) (\tau_{kj} - \tau_k) = (1) (4-1-1-1) =-1$.

Combining these results,
\begin{align*}
    \frac{1}{N} \sum_i \sum_{j \in \mathcal{N}_i} (\tau_{ki}- \tau_k) (\tau_{kj} - \tau_k) &= \frac{1}{2} (1) + \frac{1}{2} (-1) = -1/2 <0
\end{align*}

\section{Proofs} \label{sec:proof}
\subsection{Proofs for Section \ref{sec:ols_estimator}}
\begin{proof}[Proof of \Cref{thm:tau_distr}]
Observe that $\sum_{i=1}^{n_k} u_{ki}(w)=0$. Hence,
\begin{align*}
\hat{\tau}_{k}-\tau_{k} & =\frac{1}{N_{k1}}\sum_{i=1}^{n_{k}}R_{ki}W_{ki}Y_{ki}-\frac{1}{N_{k0}}\sum_{i=1}^{n_{k}}R_{ki}\left(1-W_{ki}\right)Y_{ki}-\frac{1}{n_{k}}\sum_{i=1}^{n_{k}}\left(y_{ki}(1)-y_{ki}(0)\right)\\
 & =\frac{1}{N_{k1}}\sum_{i=1}^{n_{k}}R_{ki}W_{ki}y_{ki}(1)-\frac{1}{N_{k0}}\sum_{i=1}^{n_{k}}R_{ki}\left(1-W_{ki}\right)y_{ki}(0)-\frac{1}{n_{k}}\sum_{i=1}^{n_{k}}\left(y_{ki}(1)-y_{ki}(0)\right)\\
 & =\frac{1}{N_{k1}}\sum_{i=1}^{n_{k}}R_{ki}W_{ki}u_{ki}(1)-\frac{1}{N_{k0}}\sum_{i=1}^{n_{k}}R_{ki}\left(1-W_{ki}\right)u_{ki}(0)\\
 & =\frac{1}{N_{k1}}\sum_{i=1}^{n_{k}}R_{ki}W_{ki}u_{ki}(1)-\frac{b_{k1}}{N_{k1}}\sum_{i=1}^{n_{k}}u_{ki}(1)-\frac{1}{N_{k0}}\sum_{i=1}^{n_{k}}R_{ki}\left(1-W_{ki}\right)u_{ki}(0)+\frac{b_{k0}}{N_{k0}}\sum_{i=1}^{n_{k}}u_{ki}(0)\\
 & =\frac{1}{N_{k1}}\sum_{i=1}^{n_{k}}\left(R_{ki}W_{ki}-b_{k1}\right)u_{ki}(1)-\frac{1}{N_{k0}}\sum_{i=1}^{n_{k}}\left(R_{ki}\left(1-W_{ki}\right)-b_{k0}\right)u_{ki}(0)\\
 & =\frac{n_{k}}{N_{k1}}\frac{E\left[N_{k1}\right]}{n_{k}}\frac{1}{b_{k1}n_{k}}\sum_{i=1}^{n_{k}}\left(R_{ki}W_{ki}-b_{k1}\right)u_{ki}(1)-\frac{n_{k}}{N_{k0}}\frac{E\left[N_{k0}\right]}{n_{k}}\frac{1}{b_{k0}n_{k}}\sum_{i=1}^{n_{k}}\left(R_{ki}\left(1-W_{ki}\right)-b_{k0}\right)u_{ki}(0)\\
 & =\frac{b_{k1}}{\hat{b}_{k1}}\hat{a}_{k1}-\frac{b_{k0}}{\hat{b}_{k0}}\hat{a}_{k0}
\end{align*}

where

\begin{align*}
\hat{b}_{k1} & =\frac{N_{k1}}{n_{k}}\qquad\hat{b}_{k0}=\frac{N_{k0}}{n_{k}}\\
\hat{a}_{k1} & =\frac{1}{b_{k1}n_{k}}\sum_{i=1}^{n_{k}}\left(R_{ki}W_{ki}-b_{k1}\right)u_{ki}(1)\\
\hat{a}_{k0} & =\frac{1}{b_{k0}n_{k}}\sum_{i=1}^{n_{k}}\left(R_{ki}\left(1-W_{ki}\right)-b_{k0}\right)u_{ki}(0)
\end{align*}

Since $R_{ki},W_{ki}$ are multiway clustered and well-behaved, using the law of large numbers implied by \citet{yap2023general}, it is immediate that $\frac{b_{k1}}{\hat{b}_{k1}}=1+o_{P}(1)$ and $\frac{b_{k0}}{\hat{b}_{k0}}=1+o_{P}(1)$. Then,
\begin{align*}
\hat{\tau}_{k}-\tau_{k} & =\left(1+o_{P}(1)\right)\hat{a}_{k1}-\left(1+o_{P}(1)\right)\hat{a}_{k0}\\
 & =\hat{a}_{k1}-\hat{a}_{k0}+o_{P}(1)\left(\hat{a}_{k1}-\hat{a}_{k0}\right)
\end{align*}

It is sufficient to show that $(\sqrt{N_k}/\sqrt{v_k}) (\hat{a}_{k1}-\hat{a}_{k0}) \xrightarrow{d} N(0,1)$. If that result holds, then 
\begin{align*}
\frac{\sqrt{N_k}}{\sqrt{v_k}}(\hat{\tau}_{k}-\tau_{k}) &= \frac{\sqrt{N_k}}{\sqrt{v_k}}(\hat{a}_{k1}-\hat{a}_{k0}) +o_{P}(1) \left(\hat{a}_{k1}-\hat{a}_{k0}\right)\frac{\sqrt{N_k}}{\sqrt{v_k}} \\
&= \frac{\sqrt{N_k}}{\sqrt{v_k}}(\hat{a}_{k1}-\hat{a}_{k0}) +o_{P}(1) O_P(1) \xrightarrow{d} N(0,1)
\end{align*}

Hence, 
\begin{align*}
 \hat{a}_{k1}-\hat{a}_{k0} & =\frac{1}{n_{k}}\sum_{i=1}^{n_{k}}\left(\frac{1}{b_{k1}}\left(R_{ki}W_{ki}-b_{k1}\right)u_{ki}(1)-\frac{1}{b_{k0}}\left(R_{ki}\left(1-W_{ki}\right)-b_{k0}\right)u_{ki}(0)\right)\\
 & =\frac{1}{n_{k}}\sum_{i=1}^{n_{k}}\xi_{ki} = \hat{\tau}_{k}-\tau_{k}
\end{align*}

It is obvious that $E[\xi_{ki}] =0$. Since the conditions of the central limit theorem (CLT) of \citet{yap2023general} are satisfied,
\begin{align*}
    \frac{\sum_{i=1}^{n_{k}}\xi_{ki}}{\sqrt {\sum_{i=1}^{n_k} \sum_{j \in \mathcal{N}_i} E[\xi_{ki} \xi_{kj}]}} \xrightarrow{d} N(0,1)
\end{align*}

The expression for $v_k$ is obtained by solving:
\begin{align*}
    \frac{1}{n_k^2} \sum_{i=1}^{n_k} \sum_{j \in \mathcal{N}_i} E[\xi_{ki} \xi_{kj}] = \frac{v_k}{N_k}
\end{align*}

Finally, it remains to find an expression for $E[\xi_{ki} \xi_{kj}]$. 
\footnotesize
\begin{align*}
E\left[\xi_{ki}\xi_{kj}\right] & =E [\left(\frac{1}{b_{k1}}\left(R_{ki}W_{ki}-b_{k1}\right)u_{ki}(1)-\frac{1}{b_{k0}}\left(R_{ki}\left(1-W_{ki}\right)-b_{k0}\right)u_{ki}(0)\right) \\
&\qquad \left(\frac{1}{b_{k1}}\left(R_{kj}W_{kj}-b_{k1}\right)u_{kj}(1)-\frac{1}{b_{k0}}\left(R_{kj}\left(1-W_{kj}\right)-b_{k0}\right)u_{kj}(0)\right)]\\
 & =E\left[\frac{1}{b_{k1}^{2}}\left(R_{ki}W_{ki}-b_{k1}\right)u_{ki}(1)\left(R_{kj}W_{kj}-b_{k1}\right)u_{kj}(1)-\frac{1}{b_{k1}b_{k0}}\left(R_{ki}W_{ki}-b_{k1}\right)u_{ki}(1)\left(R_{kj}\left(1-W_{kj}\right)-b_{k0}\right)u_{kj}(0)\right]\\
 & \qquad+E\left[\frac{1}{b_{k0}^{2}}\left(R_{ki}\left(1-W_{ki}\right)-b_{k0}\right)u_{ki}(0)\left(R_{kj}\left(1-W_{kj}\right)- b_{k0}\right)u_{kj}(0) \right] \\
 &\qquad - E\left[\frac{1}{b_{k1}b_{k0}}\left(R_{kj}W_{kj}-b_{k1}\right)u_{kj}(1)u_{ki}(1)\left(R_{ki}\left(1-W_{ki}\right)-b_{k0}\right)u_{ki}(0)\right]
\end{align*}
\normalsize

Expand the individual terms. 
\begin{align*}
E&\left[\frac{1}{b_{k1}^{2}}\left(R_{ki}W_{ki}-b_{k1}\right)u_{ki}(1)\left(R_{kj}W_{kj}-b_{k1}\right)u_{kj}(1)\right] \\
&=E\left[\frac{1}{b_{k1}^{2}}\left(R_{ki}W_{ki}R_{kj}W_{kj}-b_{k1}R_{kj}W_{kj}-b_{k1}R_{ki}W_{ki}+b_{k1}^{2}\right)u_{ki}(1)u_{kj}(1)\right]\\
 & =E\left[\frac{1}{b_{k1}^{2}}\left(R_{ki}W_{ki}R_{kj}W_{kj}\right)u_{ki}(1)u_{kj}(1)\right]-2\frac{E\left[R_{ki}W_{ki}\right]b_{k1}}{b_{k1}^{2}}u_{ki}(1)u_{kj}(1)+u_{ki}(1)u_{kj}(1)\\
 & =\frac{1}{b_{k1}^{2}}E\left[R_{ki}W_{ki}R_{kj}W_{kj}\right]u_{ki}(1)u_{kj}(1)-u_{ki}(1)u_{kj}(1)\\
 & =\left(\frac{1}{b_{k1}^{2}}E\left[R_{ki}W_{ki}R_{kj}W_{kj}\right]-1\right)u_{ki}(1)u_{kj}(1)
\end{align*}

Since sampling and assignment are independent,
\[
E\left[R_{ki}W_{ki}R_{kj}W_{kj}\right]=E\left[R_{ki}R_{kj}\right]E\left[W_{ki}W_{kj}\right]
\]

The remaining terms are:
\begin{align*}
E&\left[\frac{1}{b_{k1}b_{k0}}\left(R_{ki}W_{ki}-b_{k1}\right)u_{ki}(1)\left(R_{kj}\left(1-W_{kj}\right)-b_{k0}\right)u_{kj}(0)\right] \\
& =\frac{1}{b_{k1}b_{k0}}E\left[R_{ki}W_{ki}R_{kj}\left(1-W_{kj}\right)-b_{k1}b_{k0}\right]u_{ki}(1)u_{kj}(0)\\
 & =\left(\frac{1}{b_{k1}b_{k0}}E\left[R_{ki}W_{ki}R_{kj}\left(1-W_{kj}\right)\right]-1\right)u_{ki}(1)u_{kj}(0)
\end{align*}

\begin{align*}
E&\left[\frac{1}{b_{k0}^{2}}\left(R_{ki}\left(1-W_{ki}\right)-b_{k0}\right)u_{ki}(0)\left(R_{kj}\left(1-W_{kj}\right)-b_{k0}\right)u_{kj}(0)\right] \\
& =\frac{1}{b_{k0}^{2}}E\left[\left(R_{ki}\left(1-W_{ki}\right)-b_{k0}\right)\left(R_{kj}\left(1-W_{kj}\right)-b_{k0}\right)\right]u_{ki}(0)u_{kj}(0)\\
 & =\frac{1}{b_{k0}^{2}}E\left[\left(R_{ki}\left(1-W_{ki}\right)R_{kj}\left(1-W_{kj}\right)\right)-b_{k0}^{2}\right]u_{ki}(0)u_{kj}(0)\\
 & =\left(\frac{1}{b_{k0}^{2}}E\left[R_{ki}\left(1-W_{ki}\right)R_{kj}\left(1-W_{kj}\right)\right]-1\right)u_{ki}(0)u_{kj}(0)
\end{align*}

with
\begin{align*}
E\left[R_{ki}W_{ki}R_{kj}\left(1-W_{kj}\right)\right] & =E\left[R_{ki}R_{kj}\right]E\left[W_{ki}\left(1-W_{kj}\right)\right]\\
E\left[R_{ki}\left(1-W_{ki}\right)R_{kj}\left(1-W_{kj}\right)\right] & =E\left[R_{ki}R_{kj}\right]E\left[\left(1-W_{ki}\right)\left(1-W_{kj}\right)\right]
\end{align*}
The result is then obtained by plugging in these expressions.
\end{proof}

\begin{proof}[Proof of \Cref{lem:diffRW}]
If $i$ and $j$ do not share a cluster on the respective dimensions, $E\left[R_{ki}R_{kj}\right]=E[R_{ki}]E[R_{kj}]=p_{k}^{2}q_{k}^{2}$, and $E\left[W_{ki}W_{kj}\right]=\mu^{2}$. If they share the assignment dimension, then
\begin{align*}
E\left[W_{ki}W_{kj}\right] & =E\left[E\left[W_{ki}W_{kj}|A_{h(i)}\right]\right]\\
 & =E\left[A_{h}^{2}\right]\\
 & =\sigma^{2}+\mu^{2}
\end{align*}
All other derivations are analogous.
\end{proof}

\begin{proof}[Proof of \Cref{lem:mwclusW}]
\begin{align*}
E\left[W_{ki}W_{kj}\right] & =E\left[E\left[W_{ki}W_{kj}|A_{g},B_{h}\right]\right]\\
 & =E\left[A_{g}^{2}B_{h}^{2}\right]\\
 & =E\left[A_{g}^{2}\right]E\left[B_{h}^{2}\right]\\
 & =\left(\mu_{A}^{2}+\sigma_{A}^{2}\right)\left(\mu_{B}^{2}+\sigma_{B}^{2}\right)
\end{align*}
Further, 
\begin{align*}
E\left[W_{ki}W_{kj}\right] & =E\left[E\left[W_{ki}W_{kj}|A_{g},B_{h},B_{h^{\prime}},g(i)=g(j)=g,h(i)=h,h(j)=h^{\prime}\right]\right]\\
 & =E\left[A_{g}^{2}B_{h}B_{h^{\prime}}\right]\\
 & =E\left[A_{g}^{2}\right]E\left[B_{h}\right]^{2}\\
 & =\left(\mu_{A}^{2}+\sigma_{A}^{2}\right)\mu_{B}^{2}
\end{align*}
\end{proof}

\subsection{Proofs for Section \ref{sec:variance_estimators}}

\begin{proof} [Proof of \Cref{thm:var_converge}]
Due to \Cref{thm:tau_distr}, and applying a similar argument to show that $\hat{\alpha}_k/{\alpha_k} \xrightarrow{p} 1$, $\hat{U}_{ki} = U_{ki}(1+o_P(1))$. Hence, $\hat{\eta}_{ki} = \eta_{ki} (1+o_P(1))$, so it suffices to show the result for $\eta_{ki}$ directly. 

First consider the EHW result. We want to show that:
\begin{align*}
P&\left(\frac{1}{\sum_i E[\eta_{ki}^2]}\sum_{i} \left(\eta_{ki}^2  -E\left[\eta_{ki}^2\right]\right)>\epsilon \right)  \leq \frac{1}{\epsilon^{2}}\frac{1}{\left(\sum_i E[\eta_{ki}^2]\right)^{2}} E\left[\left(\sum_{i} \left(\eta_{ki}^2  -E\left[\eta_{ki}^2\right]\right)\right)^{2}\right]\\
 & = \frac{1}{\epsilon^{2}}\frac{1}{\left(\sum_i E[\eta_{ki}^2]\right)^{2}}   \sum_i \sum_{j} \left(E\left[ \eta_{ki}^2 \eta_{kj}^2 \right] -E\left[\eta_{ki}^2\right]E\left[ \eta_{kj}^2 \right]\right) \rightarrow 0
\end{align*}

Let $D_{ij}$ denote dependency between $i$ and $j$, so these units are correlated. Under our regularity conditions, for some arbitrary constant $K$, the probability above is:
\begin{align*}
    \frac{K}{n_k^2} \sum_i \sum_{j \in \mathcal{N}_i} \left(E\left[ \eta_{ki}^2 \eta_{kj}^2 \right] -E\left[\eta_{ki}^2\right]E\left[ \eta_{kj}^2 \right]\right) &= \frac{K}{n_k^2} \sum_i \sum_{j \in \mathcal{N}_i} D_{ij} \\
    &\leq \frac{K}{n_k^2} \left( \sum_g (N^G_g)^2 + \sum_h (N^H_h)^2 \right) 
\end{align*}
To obtain convergence, simply observe that:
\begin{align*}
    \frac{1}{n_k^2} \sum_g (N^G_g)^2 \leq \frac{\max_g N^G_g}{n_k} \sum_g (N^G_g) \frac{1}{n_k^{3/2}} = o(1) \frac{n_k}{n_k} = o(1) 
\end{align*}

Turning to the LZ estimator, we have $\lambda_{k}^{G} = \sum_i \sum_{j \in \mathcal{N}^G_{g_{ki}}} E[\eta_{ki} \eta_{kj}]$. We want to show that the following is $o_P(1)$, even under multi-way correlation.
\begin{align*}
P&\left(\frac{1}{\lambda_{n}^{G}}\sum_{i}\sum_{j\in\mathcal{N}_{g(i)}^{G}}\left(\eta_{ki}\eta_{kj}-E\left[\eta_{ki}\eta_{kj}\right]\right)>\epsilon\right)  \leq\frac{1}{\epsilon^{2}}\frac{1}{\left(\lambda_{n}^{G}\right)^{2}}E\left[\left(\sum_{i}\sum_{j\in\mathcal{N}_{g(i)}^{G}}\left(\eta_{ki}\eta_{kj}-E\left[\eta_{ki}\eta_{kj}\right]\right)\right)^{2}\right]\\
 & =\frac{1}{\epsilon^{2}}\frac{1}{\left(\lambda_{n}^{G}\right)^{2}}\sum_{i}\sum_{j\in\mathcal{N}_{g(i)}^{G}}\sum_{s}\sum_{l\in\mathcal{N}_{g(k)}^{G}}\left(E\left[\eta_{ki}\eta_{kj}\eta_{ks}\eta_{kl}\right]-E\left[\eta_{ki}\eta_{kj}\right]E\left[\eta_{ks}\eta_{kl}\right]\right)
\end{align*}

$E\left[\eta_{ki}\eta_{kj}\eta_{ks}\eta_{kl}\right]=E\left[\eta_{ki}\eta_{kj}\right]E\left[\eta_{ks}\eta_{kl}\right]$ whenever $(i,j)$ and $(s,l)$ do not share any cluster, even if $i,j$ share a cluster with each other. Hence, we will need to count the moment when $D_{is}=1$ or $D_{il}=1$ or $D_{js}=1$ or $D_{jl}=1$. Since moments are bounded by \Cref{asmp:regularity}, 
\[
P\left(\frac{1}{\lambda_{n}^{G}}\sum_{i}\sum_{j\in\mathcal{N}_{g(i)}^{G}}\left(\eta_{ki}\eta_{kj}-E\left[\eta_{ki}\eta_{kj}\right]\right)>\epsilon\right)\leq\frac{K_0}{\epsilon^{2}}\frac{1}{\left(\lambda_{n}^{G}\right)^{2}}\sum_{i}\sum_{j\in\mathcal{N}_{g(i)}^{G}}\sum_{s}\sum_{l\in\mathcal{N}_{g(s)}^{G}}\left(D_{is}+D_{il}+D_{js}+D_{jl}\right)
\]

It suffices to consider the sum on $D_{is}$ because everything else is analogous.

\begin{align*}
\frac{1}{\left(\lambda_{n}^{G}\right)^{2}}\sum_{i}\sum_{j\in\mathcal{N}_{g(i)}^{G}}\sum_{s}\sum_{l\in\mathcal{N}_{g(s)}^{G}}D_{is} & \leq\frac{1}{\left(\lambda_{n}^{G}\right)^{2}}\sum_{i}\sum_{j\in\mathcal{N}_{g(i)}^{G}}\left(\sum_{s\in\mathcal{N}_{g(i)}^{G}}+\sum_{s\in\mathcal{N}_{h(i)}^{H}}\right)\sum_{l\in\mathcal{N}_{g(s)}^{G}}D_{is}\\
 & \leq\frac{1}{\left(\lambda_{n}^{G}\right)^{2}}\sum_{i}\sum_{j\in\mathcal{N}_{g(i)}^{G}}\left(\sum_{s\in\mathcal{N}_{g(i)}^{G}}+\sum_{s\in\mathcal{N}_{h(i)}^{H}}\right)\max_{g}N_{g}^{G}\\
 & =\frac{1}{\left(\lambda_{n}^{G}\right)^{2}}\sum_{i}\sum_{j\in\mathcal{N}_{g(i)}^{G}}\left(N_{g(i)}^{G}+N_{h(i)}^{H}\right)\max_{g}N_{g}^{G}
\end{align*}

Taking the first term,
\begin{align*}
\frac{1}{\left(\lambda_{n}^{G}\right)^{2}}\sum_{i}\sum_{j\in\mathcal{N}_{g(i)}^{G}}N_{g(i)}^{G}\max_{g}N_{g}^{G} & =\frac{\max_{g}N_{g}^{G}}{\left(\lambda_{n}^{G}\right)^{2}}\sum_{i}\left(N_{g(i)}^{G}\right)^{2}\\
 & =\frac{\max_{g}N_{g}^{G}}{\left(\lambda_{n}^{G}\right)^{2}}\sum_{g}\left(N_{g}^{G}\right)^{3}
\end{align*}

Under \Cref{asmp:eta}, $\frac{\max_{g}\left(N_{g}^{G}\right)^{2}}{\lambda_{n}^{G}}\rightarrow0$,
and $\frac{1}{\lambda_{n}}\sum_{g}\left(N_{g}^{G}\right)^{2}$ is
finite, so

\[
\frac{\max_{g}N_{g}^{G}}{\left(\lambda_{n}^{G}\right)^{2}}\sum_{g}\left(N_{g}^{G}\right)^{3}\leq\frac{\max_{g}\left(N_{g}^{G}\right)^{2}}{\lambda_{n}^{G}}\frac{1}{\lambda_{n}^{G}}\sum_{g}\left(N_{g}^{G}\right)^{2}\rightarrow0
\]

For the other term,
\begin{align*}
\frac{1}{\left(\lambda_{n}^{G}\right)^{2}}\sum_{i}\sum_{j\in\mathcal{N}_{g(i)}^{G}}N_{h(i)}^{H}\max_{g}N_{g}^{G} & =\frac{\max_{g}N_{g}^{G}}{\left(\lambda_{n}^{G}\right)^{2}}\sum_{i}N_{g(i)}^{G}N_{h(i)}^{H}\\
 & \leq\frac{\max_{g}N_{g}^{G}\max_{h}N_{h}^{H}}{\lambda_{n}^{G}}\frac{1}{\lambda_{n}^{G}}\sum_{i}N_{g(i)}^{G}\\
 & =\frac{\max_{g}N_{g}^{G}\max_{h}N_{h}^{H}}{\lambda_{n}^{G}}\frac{1}{\lambda_{n}^{G}}\sum_{g}\left(N_{g}^{G}\right)^{2}\rightarrow0
\end{align*}

Convergence in the last line occurs due to $\max_g N^G_g/ (\lambda_n^G)^{1/2} = (\max_g (N^G_g)^2/ \lambda_n^G)^{1/2} = o(1)$. 

A similar argument works for the CGM estimator. We now want to show that the following is $o(1)$:
\begin{align*}
    \frac{K_0}{\epsilon^{2}}\frac{1}{\left( \sum_{i}\sum_{j \in \mathcal{N}_i} E[\eta_{ki} \eta_{kj}]\right)^{2}} \sum_{i}\sum_{j \in \mathcal{N}_i} \sum_{r}\sum_{l\in\mathcal{N}_s}\left(D_{is}+D_{il}+D_{js}+D_{jl}\right)
\end{align*}

As before, it suffices to consider $\sum_{i}\sum_{j \in \mathcal{N}_i} \sum_{s}\sum_{l\in\mathcal{N}_r} D_{is}$. Let $\lambda_k := \sum_{i}\sum_{j \in \mathcal{N}_i} E[\eta_{ki} \eta_{kj}]$.
\begin{align*}
    \frac{1}{\lambda_k^{2}} &\sum_{i}\sum_{j \in \mathcal{N}_i} \sum_{s}\sum_{l\in\mathcal{N}_r} D_{is} \leq \frac{1}{\lambda_k^{2}} \sum_{i}\sum_{j\in\mathcal{N}_i}\left(N_{g(i)}^{G}+N_{h(i)}^{H}\right) \left( \max_{g}N_{g}^{G} + \max_h N_h^H \right) \\
    &\leq \frac{1}{\lambda_k^{2}} \sum_{i} \left(N_{g(i)}^{G}+N_{h(i)}^{H}\right)^2 \left( \max_{g}N_{g}^{G} + \max_h N_h^H \right)
\end{align*}
It suffices to consider the leading term as the other terms are analogous. Due to \Cref{asmp:regularity},
\begin{align*}
    \frac{\max_{g}N_{g}^{G}}{\lambda_k^{2}} \sum_{i} (N_{g(i)}^{G})^2  = \frac{\max_{g}(N_{g}^{G}) ^2}{\lambda_k}  \frac{1}{\lambda_k}\sum_{g} (N_{g}^{G})^2 = o(1) O(1) = o(1)
\end{align*}

Hence, we can get asymptotic convergence to the expectation.

Since the CGM2 estimator is built from LZ estimators on the G and H dimensions, the established convergence of the LZ estimators implies the convergence of CGM2. 
\end{proof}

\begin{proof}[Proof of \Cref{cor:var_diff}]
Observe that:
\begin{align*}
    V_{k}^{CGM}-v_{k} & =\frac{N_{k}}{n_{k}^{2}}\sum_{i}\sum_{j\in\mathcal{N}_{i}}E\left[\eta_{ki}\right]E\left[\eta_{kj}\right]\\
 & =\frac{N_{k}}{n_{k}^{2}}\sum_{i}\sum_{j\in\mathcal{N}_{i}}\left(u_{ki}(1)-u_{ki}(0)\right)\left(u_{kj}(1)-u_{kj}(0)\right) 
 \end{align*}
 \begin{align*}
 V_{k}^{EHW}-v_{k} & =\frac{N_{k}}{n_{k}^{2}}\sum_{i}E\left[\eta_{ki}^{2}\right]+\frac{N_{k}}{n_{k}^{2}}\sum_{i}\sum_{j\in\mathcal{N}_{i}}E\left[\xi_{ki}\xi_{kj}\right]\\
 & =\frac{N_{k}}{n_{k}^{2}}\sum_{i}E\left[\xi_{ki}^{2}\right]+\frac{N_{k}}{n_{k}^{2}}\sum_{i}E\left[\eta_{ki}\right]^{2}-\frac{N_{k}}{n_{k}^{2}}\sum_{i}\sum_{j\in\mathcal{N}_{i}}E\left[\xi_{ki}\xi_{kj}\right]\\
 & =\frac{N_{k}}{n_{k}^{2}}\sum_{i}E\left[\eta_{ki}\right]^{2}-\frac{N_{k}}{n_{k}^{2}}\sum_{i}\sum_{j\in\mathcal{N}_{i}\backslash\left\{ i\right\} }E\left[\xi_{ki}\xi_{kj}\right] 
 \end{align*}
 \begin{align*}
 V_{k}^{LZG}-v_{k} & =\frac{N_{k}}{n_{k}^{2}}\sum_{i}\sum_{j\in\mathcal{\mathcal{N}}_{g(i)}^{G}}E\left[\eta_{ki}\eta_{kj}\right]-\frac{N_{k}}{n_{k}^{2}}\sum_{i}\sum_{j\in\mathcal{N}_{i}}E\left[\xi_{ki}\xi_{kj}\right]\\
 & =\frac{N_{k}}{n_{k}^{2}}\sum_{i}\sum_{j\in\mathcal{\mathcal{N}}_{g(i)}^{G}}E\left[\xi_{ki}\xi_{kj}\right]+\frac{N_{k}}{n_{k}^{2}}\sum_{i}\sum_{j\in\mathcal{\mathcal{N}}_{g(i)}^{G}}E\left[\eta_{ki}\right]E\left[\eta_{kj}\right]-\frac{N_{k}}{n_{k}^{2}}\sum_{i}\sum_{j\in\mathcal{N}_{i}}E\left[\xi_{ki}\xi_{kj}\right]\\
 & =-\frac{N_{k}}{n_{k}^{2}}\sum_{i}\sum_{j\in\mathcal{N}_{i}\backslash\mathcal{\mathcal{N}}_{g(i)}^{G}}E\left[\xi_{ki}\xi_{kj}\right]+\frac{N_{k}}{n_{k}^{2}}\sum_{i}\sum_{j\in\mathcal{\mathcal{N}}_{g(i)}^{G}}E\left[\eta_{ki}\right]E\left[\eta_{kj}\right]
 \end{align*}
 \begin{align*}
 V_{k}^{CGM2}-v_{k} & =\frac{N_{k}}{n_{k}^{2}}\sum_{i}\sum_{j\in\mathcal{\mathcal{N}}_{g(i)}^{G}}E\left[\xi_{ki}\xi_{kj}\right]+\frac{N_{k}}{n_{k}^{2}}\sum_{i}\sum_{j\in\mathcal{\mathcal{N}}_{g(i)}^{G}}E\left[\eta_{ki}\right]E\left[\eta_{kj}\right]-\frac{N_{k}}{n_{k}^{2}}\sum_{i}\sum_{j\in\mathcal{N}_{i}}E\left[\xi_{ki}\xi_{kj}\right]\\
 & \qquad+\frac{N_{k}}{n_{k}^{2}}\sum_{i}\sum_{j\in\mathcal{\mathcal{N}}_{h(i)}^{H}}E\left[\xi_{ki}\xi_{kj}\right]+\frac{N_{k}}{n_{k}^{2}}\sum_{i}\sum_{j\in\mathcal{\mathcal{N}}_{h(i)}^{H}}E\left[\eta_{ki}\right]E\left[\eta_{kj}\right]\\
 & =\frac{N_{k}}{n_{k}^{2}}\sum_{i}\sum_{j\in\mathcal{\mathcal{N}}_{m(i)}^{M}}E\left[\xi_{ki}\xi_{kj}\right]+\frac{N_{k}}{n_{k}^{2}}\sum_{i}\sum_{j\in\mathcal{\mathcal{N}}_{h(i)}^{H}}E\left[\eta_{ki}\right]E\left[\eta_{kj}\right]+\frac{N_{k}}{n_{k}^{2}}\sum_{i}\sum_{j\in\mathcal{\mathcal{N}}_{g(i)}^{G}}E\left[\eta_{ki}\right]E\left[\eta_{kj}\right]\geq0
\end{align*}

\end{proof}

To prove Corollary \ref{cor:cgm2pos}, I first prove Lemma \ref{lem:cgm2pos}.

\begin{lemma} \label{lem:cgm2pos}
Consider any scalar random variable $W_{i} \in \mathbb{R}$ and define $Q_n := Var(\sum_i W_i) = \sum_i \sum_{j \in \mathcal{N}_i} E [(W_i - E[W_i]) (W_j - E[W_j])]$ and $\hat{V}^w_{CGM2} := \sum_g \sum_{i,j \in \mathcal{N}^G_g} W_i W_j + \sum_h \sum_{i,j \in \mathcal{N}^H_h} W_i W_j$. Suppose $W_i \indep W_j$ when $g(i) \ne g(j)$ and $h(i) \ne h(j)$, and all the second moments exist. Then, $E[Q_n^{-1} \hat{V}^w_{CGM2}] \geq 1$.
\end{lemma}

\begin{proof} [Proof of Lemma \ref{lem:cgm2pos}]
Observe that:
 \begin{align*}
    \hat{V}^w_{CGM2} &:= \sum_g \sum_{i,j \in \mathcal{N}^G_g} W_i W_j + \sum_h \sum_{i,j \in \mathcal{N}^H_h} W_i W_j \\
    &= \sum_i \sum_{j \in \mathcal{N}_i} W_i W_j + \sum_{m} \sum_{i,j \in \mathcal{N}^{M}_{m}} W_i W_j
\end{align*}

Define:
\begin{align*}
    Q^*_n := \sum_i \sum_{j \in \mathcal{N}_i} E[W_i] E[W_j]
\end{align*}


Then,
\small
\begin{align*}
    E[Q_n^{-1} \hat{V}^w_{CGM2}] &= Q_n^{-1} E\left( \sum_i \sum_{j \in \mathcal{N}_i} (W_i - E[W_i]) (W_j - E[W_j]) +  \sum_i \sum_{j \in \mathcal{N}_i} (2W_i E[W_j] -  E[W_i] E[W_j]) + \sum_{m} \sum_{i,j \in \mathcal{N}^{M}_{m}} W_i W_j \right) \\
    &= 1  + Q_n^{-1} E\left(Q^*_n  + \sum_{m} \sum_{i,j \in \mathcal{N}^{M}_{m}} W_i W_j \right)
\end{align*}

\normalsize
Then, it remains to show that $ \left(Q^*_n  + \sum_{m} \sum_{i,j \in \mathcal{N}^{M}_{m}} E[W_i W_j] \right)$ is weakly positive. 
\begin{align*}
    Q^*_n &+ \sum_{m} \sum_{i,j \in \mathcal{N}^{M}_{m}} E[W_i W_j] = \sum_i \sum_{j \in \mathcal{N}_i} E[W_i] E[W_j] + \sum_{m} \sum_{i,j \in \mathcal{N}^{M}_{m}} E[W_i W_j] \\
    &= \sum_g \sum_{i,j \in \mathcal{N}^G_g} E[W_i] E[W_j] + \sum_h \sum_{i,j \in \mathcal{N}^H_h} E[W_i] E[W_j] + \sum_{m} \sum_{i,j \in \mathcal{N}^{M}_{m}} (E[W_i W_j] - E[W_i] E[W_j])
\end{align*}

Terms like $\sum_g \sum_{i,j \in \mathcal{N}^G_g} E[W_i] E[W_j]$ are positive definite. Let $W_g \in \mathbb{R}^{N^G_g}$ be the vector that stacks the $W_i$'s in cluster $g$. Then,
\begin{align*}
    \sum_g \sum_{i,j \in \mathcal{N}^G_g} E[W_i] E[W_j] &= \sum_g 1_{N^G_g}' E[W_g] E[W_g]' 1_{N^G_g}
\end{align*}
Since $E[W_g] E[W_g]'$ is positive definite by construction of an outer product, $1_{N^G_g}' E[W_g] E[W_g]' 1_{N^G_g} \geq 0$. Thus, it comes down to $\sum_{m} \sum_{i,j \in \mathcal{N}^{M}_{m}} (E[W_i W_j] - E[W_i] E[W_j])$. 
Observe that:
\begin{align*}
    E [(W_i - E[W_i]) (W_j - E[W_j])] &= E[W_i W_j] - 2 E[W_i] E[W_j] + E[W_i] E[W_j] \\
    &=  E[W_i W_j] - E[W_i] E[W_j]
\end{align*}

Thus,
\begin{align*}
    \sum_{m} \sum_{i,j \in \mathcal{N}^{M}_{m}} (E[W_i W_j] - E[W_i] E[W_j]) &= \sum_{m} \sum_{i,j \in \mathcal{N}^{M}_{m}} E[(W_i - E[W_i]) (W_j - E[W_j])]
\end{align*}

Applying a similar argument, 
\begin{align*}
    \sum_{m} \sum_{i,j \in \mathcal{N}^{M}_{m}} E[(W_i - E[W_i]) (W_j - E[W_j])] &= \sum_{m} E [1_{N_{m}}' (W_{m} - E[W_{m}]) (W_{m} - E[W_{m}])' 1_{N_{m}}] \geq 0
\end{align*}

Thus, $ \left(Q^*_n  + \sum_{m} \sum_{i,j \in \mathcal{N}^{M}_{m}} E[W_i W_j] \right) \geq 0$, yielding the result.

\end{proof}

\begin{proof} [Proof of Corollary \ref{cor:cgm2pos}]
Following the proof of Theorem \ref{thm:var_converge}, $\hat{\eta}_{ki} = \eta_{ki} (1 + o_P(1))$, so it suffices to work with $\eta_{ki}$. By using $\eta_{ki}$ in place of $W_i$ in Lemma \ref{lem:cgm2pos}, $V_k^{CGM2}/ v_k \geq 1$, so $V_k^{CGM2} \geq v_k$.
\end{proof}

\bibliographystyle{ecta}
\bibliography{mwclus}

\end{document}